\documentclass[a4paper,11pt,DIV=calc]{article}

\usepackage[a4paper,left=30mm,right=30mm,top=30mm,bottom=30mm,marginpar=23mm]{geometry} 

\usepackage{hyperref}
\usepackage{subfig}
\usepackage{amsmath,amssymb}
\usepackage{bbm}
\usepackage{amsthm}
\usepackage{pgfplots}
\usepackage{graphicx}
\usepackage{MnSymbol}
\usepackage{color}
 \theoremstyle{plain} \newtheorem{Theorem}{Theorem}
\theoremstyle{plain} 
\theoremstyle{plain} \newtheorem{Lemma}{Lemma}
\theoremstyle{plain} 
\theoremstyle{definition} \newtheorem{Def}{Definition}
\theoremstyle{definition}\newtheorem{Remark}{Remark}
\newcommand{\dist}{\operatorname{dist}}

\newcommand{\id}{{\mathbb{I}}}

\newcommand{\eps}{\varepsilon}

\newcommand{\rk}{\operatorname*{rank}}

\newcommand{\onn}{\mbox{ on }}

\newcommand{\fall}{\mbox{ for all }}

\newcommand{\om}{\Omega}

\newcommand{\so}{{\rm{SO(3)}}}

\newcommand{\rr}{\mathbb{R}^{3\times 3}}
\newcommand{\R}{{\mathbb{R}^3}}

\newcommand{\tp}{\otimes}

\mathchardef\mdash="2D

\newcommand{\PP}{\mathcal{P}^{24}}

\newcommand{\qc}{{\mbox{\scriptsize{qc}}}}
\usepackage{framed}
\usepackage{authblk}

\usepackage{fancyhdr}

\newcommand\shorttitle{Interior nucleation of Martensite}
\newcommand\authors{A. Muehlemann}

\begin{document}

\makeatletter
\def\@maketitle{%
  \newpage
  \null
  \vskip 2em%
  \begin{center}%
  \let \footnote \thanks
    {\Large\bf  \@title \par}%
    \vskip 1.5em%
    {\normalsize
      \lineskip .5em%
      \begin{tabular}[t]{c}%
        \@author
      \end{tabular}\par}%
    \vskip 1em%
    {\normalsize \@date}%
  \end{center}%
  \par
  \vskip 1.5em}
\makeatother
\title{{Interior Nucleation of Martensite in a Cubic-to-Tetragonal Phase Transformation}}

\author{Anton Muehlemann%
  \thanks{\texttt{muehlemann@maths.ox.ac.uk}}}
\affil{\small\textit{Mathematical Institute, University of Oxford,}\\\textit{ Andrew Wiles Building, Radcliffe Observatory Quarter, Woodstock Road,} \\ \textit{Oxford OX2 6GG, United Kingdom}}

\date{Dated: \today}

\maketitle

%
%
%
%
%
 \hypersetup{
  pdfauthor = {Anton Muehlemann (University of Oxford)},
  pdftitle = {Interior Nucleation of Martensite in a Cubic-to-Tetragonal Phase Transformation},
  pdfsubject = {MSC (2010): 74N05, 74N15},
  pdfkeywords = {lath martensite, microstructure, twins within twins, (557) habit planes, double shear theories, energy minimisation, non-classical interfaces}
}

\begin{abstract}
\noindent Using a variational model based on non-linear elasticity we investigate whether in a cubic-to-tetragonal phase transformation it is energetically preferable to nucleate martensite within austenite. Under minimal growth assumptions on the free energy density $W$ away from the wells, we derive explicit upper bounds on $W^{qc}({\mathbb{I}})$, i.e. on the macroscopic free energy density of a region that has been macroscopically deformed by the identity map. The bounds only depend on material parameters, the temperature difference and the growth of $W$ away from the wells. By comparing $W^{qc}({\mathbb{I}})$ and $W({\mathbb{I}})$ we conclude that nucleation is always energetically preferable and are able to provide quantitative upper bounds for the gain in negative energy due to nucleation. 
\vspace{4pt}



\vspace{4pt}
\end{abstract}

\section{Introduction}
The high growth velocity of heterogeneous temperature induced  nucleation of martensite, which is of the order of magnitude of sound (\cite{Bunshah}, \cite[p.977~ff.]{Christian}), is one of the reasons why the process of nucleation is not yet satisfactory understood. As pointed out in \cite[Ch.~23]{Christian} ``the kinetics of martensitic transformations are depended mainly on the nucleation process as each plate grows rapidly to its limiting size'' and as a consequence ``it is often found that the plates of martensite formed (...) have a constant lenght:thickness ratio''  until ``growth ceases because of plastic deformation''. Different models (e.g. \cite{Kostas,Cook,OlsonDislo,OlsonMartNuc}) have been proposed to explain the formation and predict properties of such plates. It has been shown theoretically and experimentally that the availability of a free surface can have a significant influence on the nucleation process. In \cite{Krauss} a single crystal of Fe-Ni was coated 
with a Ni-rich surface layer to inhibit nucleation from the free surface with the effect that the formation of martensite was completely suppressed during slow cooling. Applying a similar procedure to an In-Cd alloy however did not show a difference between the coated and uncoated crystal. In a related experiment on Cu-Al-Ni single- and polycrystals with varying grain-sizes it has been shown in \cite{Salz} that the energy needed for bulk nucleation is significantly higher than that required for corner nucleation or to move a planar interface. A single crystal of the same alloy was considered in \cite{Kostas} where it was shown experimentally and theoretically that corner nucleation is energetically preferable over interior nucleation. However if a free surface is not available, for instance because of a relatively large grain-size or by a coating as described above, it is of interest to establish criteria for interior nucleation. 

The present paper investigates the energetic properties of interior nucleation using a variational model based on non-linear elasticity first introduced in \cite{BallFine}. In contrast to \cite{OlsonMartNuc}, where a nucleus forms with an array of interfacial dislocation, in this model the material is assumed to deform purely elastically without considering plasticity. Further interfacial energy is not taken into account. This is justifiable by recent experiments in \cite{James} on Ti-Ni-Pd and Ti-Ni-Au crystals with very low hysteresis where it has been observed that structures with very complex interfaces form, even though less complex interfaces are already kinematically compatible, so that the contribution from interfacial energy must be negligible.

A particularly interesting case of interior nucleation is self-accommodation. Using the Ball-James (\cite{BallFine}) model it has been shown in \cite{Bhaself} that a necessary and sufficient condition for self-accommodation is that the austenite and martensite have equal volume. 
In this case the martensite can 
rearrange itself such that it matches the austenite along any surface and thus interior nucleation is energetically as preferable as nucleation from the boundary. However in general we expect to have a volume change during transformation, so that any inclusion of martensite within austenite necessarily has a region where it is neither on the austenite nor on the martensite wells. We will refer to this region as the interpolation layer and it can loosely be understood as the elastic equivalent to having interfacial energy and plasticity. By the classical theory of heterogeneous nucleation in order to transform a given region of austenite to martensite ``the negative free energy change resulting from the formation of a given volume of a more stable phase [the martensite] is opposed by a positive free energy change due to the creation of an interface [the interpolation layer] between the initial phase and the new phase'' (\cite[p.~4]{Christian}). Thus our criterion for interior nucleation is that the energy of the 
inclusion is lower than the 
energy of the parent phase, i.e. the austenite. It seems to have first been shown qualitatively in \cite{BallSummer} that for a sufficiently flat cylindrical inclusion nucleation is always energetically preferable over the pure austenite. In this paper we develop quantitative estimates on the energy difference between a nucleus and the austenite depending on the temperature and the assumed growth of the free energy density away from the wells.

\begin{framed} 
\noindent{\bf{Nomenclature}}
\begin{align*}
&W^{1,\infty}(\Omega) &&\mbox{Space of continuous functions from $\om$ to $\R$ with bounded derivative.}\\
&W^{1,\infty}_0(\Omega) &&\mbox{$W^{1,\infty}(\Omega)$ with zero boundary conditions.}\\
&\rr &&\mbox{Space of $3\times 3$ matrices.}\\
&\so &&\mbox{Space of rotation matrices, i.e. $R\in \rr$ such that $R^TR=\id, \ \det R=1$.}\\
&\PP &&\mbox{Crystallographic point group of a cube.}
\end{align*}
\end{framed}
\section{Description of the Model}
Following \cite{BallFine} microstructures correspond to minimising sequences of the total free energy
\begin{equation}
 E(y):= \int_\Omega W_\theta(Dy(x)) dx, \label{FreeEnergy} \tag{E}
\end{equation}
where $\Omega \subset \mathbb{R}^3$ represents the undistorted austenite at temperature $\theta$ and $y \in W^{1,\infty}(\Omega)$ maps a point $x$ in the reference
configuration to a point $y(x)$ in the deformed configuration. In this continuum model the mesoscopic free energy density $W$ depends only on the deformation gradient $F=Dy(x)$ and temperature $\theta$. The deformation $y(x)=x, \, Dy(x)={\mathbb{I}}$ corresponds to pure austenite and
thus, since we are interested in modelling inclusions of martensite within austenite, $y$ has to satisfy the homogeneous
displacement boundary condition $y(x)=x$ on $\partial \Omega$. Furthermore $W$ respects frame indifference, i.e. $W(RA)=W(A) \, \fall R \in SO(3), \, A \in \mathbb{R}^{3\times 3}_+$ and the cubic symmetry of the austenite, i.e. $W(AQ)= W(A) \fall Q \in \mathcal{P}^{24}$. Let $K_m={\rm{SO(3)}} U_1 \mathcal{P}^{24}=\bigcup_{i=1}^3{\rm{SO(3)}} U_i$, $U_1=\operatorname{diag}(\eta_2,\eta_1,\eta_1),U_2=\operatorname{diag}(\eta_1,\eta_2,\eta_1),U_3=\operatorname{diag}(\eta_1,\eta_1,\eta_2)$ denote the set of martensitic wells. Then for a fixed temperature $\theta^*$ below the transition temperature we assume 
\begin{equation}
  W(R)=0 \mbox{ for } R\in \so, \  W(U)=-\tau \mbox{ for } U \in K_m,  \ W>-\tau \onn \mathbb{R}^{3\times 3} \backslash K_m. \label{W}
 \end{equation}
Clearly the energy functional \eqref{FreeEnergy} is bounded from below by $E(y) \geq - \tau \operatorname{Vol}(\Omega).$ Hence any sequence $\lbrace y_j \rbrace \subset W^{1,\infty}(\Omega)$ such that $E(y_j)\rightarrow-\tau \operatorname{Vol}(\Omega)$ corresponds to a microstructure that can only contain variants of martensite and thus can only exist if the volume does not change during transformation. If the transformation does involve a change in volume, i.e. $\det U_1 \neq 1$, it is of interest to find the minimal value of the free energy $E$. Since the energy is proportional to the volume of the inclusion the relevant quantity is the macroscopic free energy density, i.e. the average energy per volume element, denoted by $W^{qc}$. $W^{qc}({\mathbb{I}})$ is the free energy density of a region that has been macroscopically deformed by the identity map, i.e. that obeys the boundary condition $y(x)=x$. A characterisation (see  \cite{Dacqc}) of $W^{qc}$ is given by
 \begin{equation*} 
   W^{qc}({\mathbb{I}}) = \frac{1}{\operatorname{Vol}(\Omega)} \inf_{\phi \in W_0^{1,\infty}(\Omega)} \int_\Omega W ({\mathbb{I}} + D\phi(x))dx, \label{Wqc}
\end{equation*}
so that $\inf E(y)=\operatorname{Vol}(\Omega)W^{qc}({\mathbb{I}})$. In particular whenever $W^{qc}({\mathbb{I}}) <0$ there exists a deformation $y(x)=x+\phi(x)$ that is energetically preferable over the pure austenite.

\section{Interior Nucleation of Martensite}
In the following we will always assume $\det U_1 \neq 1$ so that we cannot have stress-free inclusions. The goal is to construct a test function $\phi$ such that the corresponding deformation $y(x)=x+\phi(x)$ has negative total free energy and is thus  energetically preferable over the pure austenite corresponding to $y(x)=x$. We will explore two different approaches. In the first approach we only allow deviations from the austenite and we show that the geometry of the inclusion plays a crucial r\^{o}le. In the second approach we only allow deviations from the martensite wells and use the principle of self-accommodation to derive our estimates. 

\subsection*{Deviations from the Austenite}
By \cite{BallFine} it is known that a fine mixture of two rank-one connected martensitic variants can form a stress-free planar interface with the austenite. That is, there exists a rank-one matrix $\mathbf{b} \tp \mathbf{m}, \ |\mathbf{m}|=1$ such that 
\begin{equation} 
 \lambda \hat R U_1 + (1-\lambda)\hat R RU_2 ={\mathbb{I}}+\mathbf{b} \tp \mathbf{m},\label{Comp}
\end{equation}
where $\hat R U_1, \hat R RU_2\in K_m$ such that $U_1-RU_2=\mathbf{a} \tp \mathbf{m},\ |\mathbf{n}|=1$ and $\lambda\in (0,1)$. We consider an inclusion with a core of pure martensite. Since the volumes do not agree we have a transition layer where the deformation gradients cannot be supported on the wells. We show that the energy penalty from the transition layer can be overcome by the gain in negative energy due to the presence of the martensite as long as the inclusion is flat and thin.

\begin{Def}{(Cylindrically symmetric inclusion with a simple laminate core)} \label{DefEllipsInc}\\
 Let $\Omega=\Omega_1 \cup \Omega_2$ be such that the core $\Omega_1$ consists of a simple laminate of martensite that is compatible with the austenite in the sense of \eqref{Comp}. By rotating the entire system we may assume without loss of generality  $\mathbf{m}=e_3$. We define the two regions as follows
\begin{equation*}  
  \Omega_1:=\lbrace x\in \mathbb{R}^3 | x_3 \in [0,\varphi(r)] \rbrace ,\, \Omega_2:=\lbrace x\in \mathbb{R}^3 | x_3 \in (\varphi(r),\sigma \varphi(r)] \rbrace,
\end{equation*}
where $r^2=x_1^2+x_2^2, \ \sigma>1$ and $0 \leq \varphi(r) \in W^{1,\infty}([0,\infty)), \, \lim_{r \rightarrow \infty }\varphi(r)=0$.
  \begin{figure}[h]
  \centering
  \includegraphics[width=10cm]{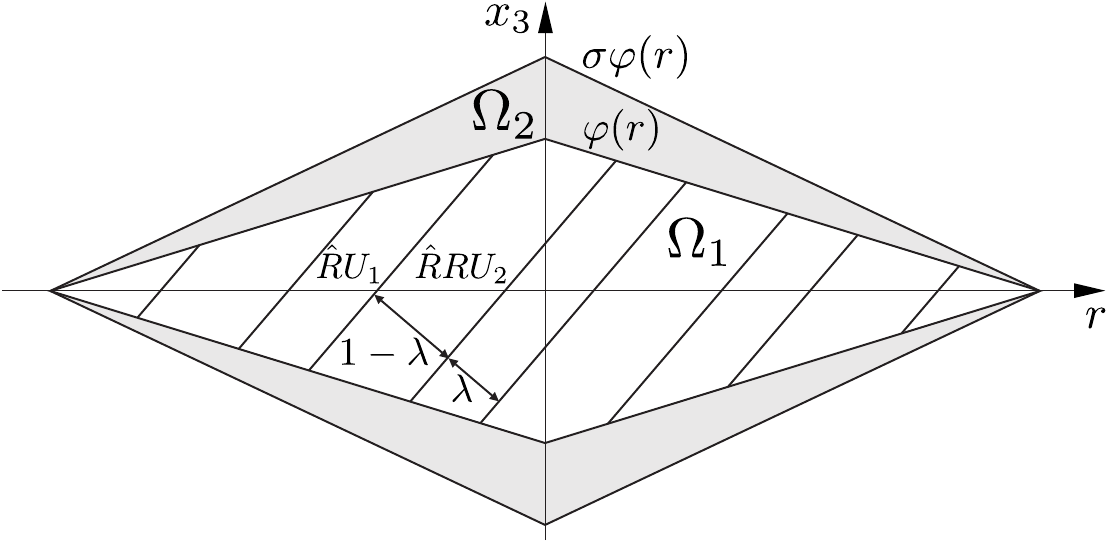}
  \caption{Cross section of the inclusion at $x_2=0$ in the $x_1-x_3$ plane}
  \end{figure}
Note that we assume symmetry with respect to the $x_1-x_2$ plane and only define the upper half of the inclusion. We further assume that $\operatorname{Vol}(\Omega_1)=2 \pi\int_0^\infty  r \varphi(r) dr=1$ and thus $\operatorname{Vol}(\Omega_2)=\sigma-1$. We define our test function $\phi\in W_0^{1,\infty}(\Omega)$ in the regions $\Omega_1$ and $\Omega_2$ separately. In $\Omega_1$ we set $ \phi_1(x)\equiv\left. \phi \right|_{\Omega_1} :=\mathbf{b}\cdot x_3$ and in $\Omega_2$ we interpolate linearly in $x_3$ direction from $\mathbf{b}\cdot \varphi(r)$ to $0$ by defining
\begin{align*}  
  \phi_2 (x) \equiv\left. \phi \right|_{\Omega_2}:=\frac{\mathbf{b}}{1-\sigma} \cdot \left(x_3-\sigma \varphi \left(r\right)\right),
\end{align*}
so that $\phi=\phi_1+\phi_2 \in W_0^{1,\infty}(\Omega)$. 
\end{Def}

\begin{Theorem} \label{ThSimple}
Let $W$ be as in \eqref{W} and assume the following upper bound on $W$ away from the austenite  
\begin{equation} 
 W({\mathbb{I}}+H) \leq C \left| H\right|^2. \label{WupperBound}
\end{equation}
Then the inclusion given by Definition \ref{DefEllipsInc} provides the following upper bound 
\begin{equation} 
 W^\qc({\mathbb{I}}) \leq -\tau \left(1+2\beta + 2 \sqrt{\beta+\beta^2}\right)^{-1}, \label{LamBound}
\end{equation}
where $\beta:=\frac{C}{\tau}|\mathbf{b}|^2$. 
\end{Theorem}
\begin{proof}
In $\Omega_1$ the gradients of $\phi$ are supported on the martensite wells and therefore $E(\phi_1)= -\tau \operatorname{Vol}(\Omega_1)=-\tau$. In $\Omega_2$ we estimate the energy using the upper bound on $W({\mathbb{I}}+H)$. We have
$D\phi_2(x)=\frac{\mathbf{b}}{1-\sigma} \tp \nabla u(x)$ with $\nabla u (x) = \left(-\sigma \varphi'(r) \frac{x_1}{r},-\sigma \varphi'(r) \frac{x_2}{r},1\right)^T.$ Noting that $D\phi_2(x)$ is independent of $x_3$ and using the coarea formula (see \cite[Ch.~3]{Federer}) we arrive at
\begin{align*}  
 E(\phi_2) \leq \frac{C}{(\sigma-1)} |\mathbf{b}|^2 + \frac{\sigma^2}{(\sigma-1)} \gamma, 
\end{align*}
where $\gamma:={C \pi}|\mathbf{b}|^2 \int_0^\infty r \varphi(r)\varphi'(r)^2 dr$. Combining the two estimates we obtain
\begin{equation*} 
 E(\phi) \leq - \tau + \frac{C}{(\sigma-1)} |\mathbf{b}|^2 + \frac{\sigma^2}{(\sigma-1)} \gamma
\end{equation*}
and thus for $\gamma \rightarrow 0$ and since $\operatorname{Vol}(\Omega)=\sigma$ we have
\begin{equation*} 
 W^{qc}({\mathbb{I}}) \leq \tau \left(- \frac{1}{\sigma}+\frac{C}{\tau}\frac{1}{\sigma(\sigma-1)} |\mathbf{b}|^2 \right)= \tau \left(- \frac{1}{\sigma}+\frac{\beta}{\sigma(\sigma-1)} \right),
\end{equation*}
where $\beta:=\frac{C}{\tau}|\mathbf{b}|^2$. Minimisation with respect to $\sigma$ gives the desired bound.
\end{proof}
\begin{Remark}
 We note that we derived this upper bound by taking the limit $\gamma\rightarrow0$ and thus $\int_0^\infty r \varphi(r)\varphi'(r)^2 dr\rightarrow0$ which corresponds to $\varphi'(r) \rightarrow 0$ as the volume is fixed. Thus within the class of inclusions under consideration the lowest energy is achieved by a very flat and thin inclusion. 
\end{Remark} 
We can derive the same upper bound on $W^{qc}({\mathbb{I}})$ by estimating $W^{rc}({\mathbb{I}})$, where $W^{rc}$ denotes the rank-one convex envelope of $W$ (see e.g. \cite{Daco}). The proof uses the general result that $W^{qc} \leq W^{rc}$ and is simpler than the proof of Theorem \ref{ThSimple} but gives no information on the structure of the inclusion. 
\begin{Lemma}
Under the same assumptions as in Theorem \ref{ThSimple} an upper bound on $W^{rc}({\mathbb{I}})$ and thus on $W^{qc}({\mathbb{I}})$ is given by
\begin{equation*} 
 W^{rc}({\mathbb{I}}) \leq -\tau \left(1+2\beta + 2 \sqrt{\beta+\beta^2}\right)^{-1},
\end{equation*}
where $\beta := \frac{C}{\tau} |\mathbf{b}|^2$.
\end{Lemma}
\begin{proof}
An upper bound on $W^{rc}$ (see \cite[Thm.~6.10]{Daco}) is in its simplest form given by 
\begin{equation} \label{wrc}
W^{rc}(A) \leq \inf \lbrace \lambda W(A_1)+(1-\lambda) W(A_2): \lambda A_1+(1-\lambda)A_2=A, \rk(A_1-A_2)\leq 1 \rbrace.
\end{equation}
The estimate can be refined by applying the same estimate to each of the terms $W(A_i),i=1,2$ on the RHS. We now decompose ${\mathbb{I}} = \lambda ({\mathbb{I}}+\mathbf{b} \tp \mathbf{m})+ (1-\lambda) ({\mathbb{I}} + \eps b \tp m), \ \eps = \frac{\lambda}{\lambda-1}$ and $ {\mathbb{I}}+\mathbf{b} \tp \mathbf{m}= \lambda R_1U_1+(1-\lambda) R_2 U_2, \ R_1,R_2\in {\rm{SO(3)}}$ such that $\rk (R_1U_1-R_2U_2)=1$. 
Then by using \eqref{WupperBound} and \eqref{wrc} we have
\begin{align*} 
W^{rc}({\mathbb{I}})\leq \lambda W^{rc}({\mathbb{I}}+\mathbf{b} \tp \mathbf{m})+(1-\lambda)W({\mathbb{I}}+\frac{\lambda}{\lambda-1} \mathbf{b} \tp \mathbf{m}) 
\leq  \tau\left(- \lambda + \beta \frac{\lambda^2}{1-\lambda}\right) 
\end{align*}
with $\beta := \frac{C}{\tau} |\mathbf{b}|^2$. Minimisation with respect to $\lambda$ gives the desired estimate.
\end{proof}

\subsection*{Deviations from the Martensite}
In many materials $\det U_1\approx 1$, so that matrices in a relatively small neighbourhood of the martensite wells have unit determinant.
Allowing deviations from the martensite wells we now consider an inclusion with a core of variants that are close to the martensite
wells and allow self-accommodation and thus do not require a transition layer. 

\begin{Theorem}\label{TheoremUpper}
 Let $|\det U^{-1/3} -1|\ll |\eta_1-1|, |\eta_2-1|, \, U \in K_m=\bigcup_{i=1}^3{\rm{SO(3)}} U_i$ and let $W$ be as \eqref{W} with
 \begin{equation} 
 W(V) \leq W(U)+{C} \dist^2(V, K_m) \label{MartExp}
\end{equation}
for all $V\in \mathbb{R}^{3\times 3}$ such that $\dist(V,K_m)\leq |1-\det U^{-1/3}|$. Then
\begin{equation} 
  W^{qc}({\mathbb{I}})\leq W^{rc}({\mathbb{I}})\leq-\tau+C(\det U^{-1/3}-1)^2 |U|^2.\label{Upperbound}
\end{equation}
\end{Theorem}
\begin{proof}
By Bhattacharya's construction (see \cite{Bhaself}) in a cubic-to-tetragonal phase transformation there exist laminates with macroscopic deformation gradient ${\mathbb{I}}$ iff $\det  U_1=1$. Therefore let us set $\tilde U_i := (\det U)^{-1/3} U_i$ so that $\det \tilde U_i=1$. By \eqref{MartExp} 
we have $$ W(\tilde U_i) \leq - \tau + C \left( 1-\det U^{-1/3}\right)^2 |U|^2$$ and thus using \eqref{wrc} we can estimate
\begin{align*} 
 W^{qc}({\mathbb{I}}) \leq W^{rc}({\mathbb{I}}) \leq W(U)+C \left( 1-\det U^{-1/3}\right)^2 |U|^2. 
\end{align*}
\end{proof}
\begin{Remark}
We note that unlike Theorem \ref{ThSimple} this construction is only preferable over the pure austenite for $\tau$ sufficiently large, corresponding to a temperature $\theta^*$ sufficiently far below the transition temperature $\theta_c$, or for $\det U$ sufficiently close to $1$. A comparison with the estimates in Theorem \ref{ThSimple} reveals that \eqref{Upperbound} provides a better bound than \eqref{LamBound} for small volume changes and a worse bound on $W^{qc}({\mathbb{I}})$ for large volume changes. For $\det U \rightarrow 1$ we recover the case of self-accommodation, i.e. 
 \begin{equation*} 
  -\tau \leq W^{qc}({\mathbb{I}})\leq -\tau +\lim_{\det U \rightarrow 1}C(\det U^{-1/3}-1)^2 |U|^2=-\tau.
 \end{equation*}
\end{Remark}
\begin{Remark}
For simplicity the present paper only considered a cubic-to-tetragonal phase transformation. However, in the first part the only assumption on the material was the existence of an interface between a microstructure of martensite and austenite. In the second part, the only assumption on the material was that it almost satisfies a criterion of self-accommodation (e.g. in a cubic-to-* transformation this holds iff $\det U_i \approx 1$, cf. \cite{Bhaself}). Thus as long as either or both assumptions are satisfied all respective conclusions remain the same. 
\end{Remark}

\section{Conclusions}
Without making any a priori assumptions, apart from cylindrical symmetry, on the shape of an inclusion of martensite within austenite we have shown that - within the class of inclusions under consideration - the lowest energy is achieved for a flat and thin shape as observed in experiments. We showed that $W^{qc}({\mathbb{I}})$ is always negative, so that within the model, it is always energetically preferable to nucleate martensite within austenite. Further by comparing two different types of nucleation in Theorem \ref{ThSimple} and Theorem \ref{TheoremUpper} we could show that for materials that almost fulfil the criteria for self-accommodation, i.e. $\det U_1 \approx 1$, it is energetically preferable to form a self-accommodating inclusion rather than an inclusion with a core of a simple laminate. 

\section*{Acknowledgements}
The research leading to these results has received funding from the European Research Council under the European Union's Seventh Framework Programme (FP7/2007-2013) / ERC grant agreement n$^{\circ} \, 291053$.

\bibliographystyle{alpha}

\end{document}